\documentclass[10pt,journal,compsoc]{IEEEtran}

\ifCLASSINFOpdf
\usepackage[pdftex]{graphicx}
  \else
  \fi
\usepackage{cite}

\usepackage{color}
\usepackage[ruled,linesnumbered]{algorithm2e}
\usepackage{amssymb,amsmath}
\usepackage{amsmath}
\usepackage{mathrsfs}
\usepackage{amsmath,amsthm,amssymb,amsfonts}
\usepackage{multirow}
\usepackage{booktabs}
\usepackage{array}
\usepackage{makecell}
\usepackage{stfloats}

\usepackage{epstopdf}
\usepackage{epsfig}

\makeatletter
\thm@headfont{\sc}
\makeatother
\newtheorem{theorem}{Theorem}

\begin{document}

\title{Target Defense Against Link-Prediction-Based Attacks via Evolutionary Perturbations}

\author{Shanqing Yu, Minghao Zhao, Chenbo Fu, Huimin Huang, Xincheng Shu, Qi Xuan,~\IEEEmembership{Member,~IEEE}, \\and Guanrong Chen,~\IEEEmembership{Fellow,~IEEE}
\IEEEcompsocitemizethanks{
\IEEEcompsocthanksitem S. Yu, M. Zhao, C. Fu, H. Huang, X. Shu, and Q. Xuan are with the College of Information Engineering, Zhejiang University of Technology, Hangzhou 310023, China (e-mail: yushanqing@zjut.edu.cn, yzbyzmh1314@163.com, cbfu@zjut.edu.cn,  saffronHuang@163.com, sxc.shuxincheng@foxmail.com, xuanqi@zjut.edu.cn).
\IEEEcompsocthanksitem G. Chen is with the Department of Electronic Engineering, City University of Hong Kong, Hong Kong SAR, China (e-mail: eegchen@cityu.edu.hk).
\IEEEcompsocthanksitem This article has been submitted on June 27th, 2018.
\IEEEcompsocthanksitem Corresponding author: Qi Xuan.
}
}

\IEEEtitleabstractindextext{
\begin{abstract}
In social networks, by removing some target-sensitive links, privacy protection might be achieved.
However, some hidden links can still be re-observed by link prediction methods on observable networks.
In this paper, the conventional link prediction method named Resource Allocation Index (RA) is adopted
for privacy attacks. Several defense methods are proposed, including heuristic and evolutionary approaches,
to protect targeted links from RA attacks via evolutionary perturbations. This is the first time to study
privacy protection on targeted links against link-prediction-based attacks. Some links are randomly selected from the network as targeted links
for experimentation. The simulation results on six real-world networks demonstrate the superiority of
the evolutionary perturbation approach for target defense against RA attacks. Moreover, transferring
experiments show that, although the evolutionary perturbation approach is designed to against RA attacks,
it is also effective against other link-prediction-based attacks.
\end{abstract}

\begin{IEEEkeywords}
social network, link prediction, resource allocation index, target defense, evolutionary algorithm,  transferability.
\end{IEEEkeywords}}

\maketitle

\IEEEdisplaynontitleabstractindextext

\IEEEpeerreviewmaketitle

\section{Introduction}

\IEEEPARstart{M}{any} complex systems in the real world can be represented by networks, where the nodes
denote the entities in the real systems and links capture various relationships among them~\cite{newman2003structure,Boccaletti2006Complex,strogatz2001exploring,jeong2001lethality}.

In network science, \emph{link prediction\/} is a fundamental notion, which attempts to uncover missing
links or to predict future interactions between pairwise nodes, based on observable links and other
external information. Link prediction has various applications in many fields, e.g., link prediction
itself can serve for network reconstruction~\cite{Guimer2009Missing}, evaluating network evolving
mechanism~\cite{Chengdu2011Emergence} and node classification~\cite{Gallagher2008Using}. Besides,
in protein-protein interaction networks, link prediction can serve to guide the design of experiments
to find previously undiscovered interactions~\cite{Mering2005STRING}. With the boom of social media
in recent years, service providers can leverage social networks and other information to recommend
friends, commodities and advertisement that may attract
users~\cite{lu2012recommender,resnick1997recommender,chen2005link,song2009scalable,cukierski2011graph}.
This process can also be naturally described as link prediction.

While various social media are enriching people's lives~\cite{Wu2014Data,xuan2018modern,Xuan2018Social}, there is an increasing concern about privacy
issues since more and more personal information could be obtained by others online. In the field of
social privacy protection, many algorithms have been developed for protecting the privacy of users,
such as identity, relationship and attributes, from different situations in which different public
information was exposed to adversaries~\cite{Liu2008Towards,Cheng2010K,Tai2011Privacy,zheleva2008preserving}.
In this paper, the focus is on preserving link privacy in social networks.

In retrospect, Zheleva \textit{et al.}~\cite{zheleva2008preserving} proposed the concept of link
re-identification attack, which refers to inferring sensitive relationships from anonymized network data.
If the sensitive links can be identified by the released data, then this means privacy breach.
In early days, link perturbation was considered a common technique to preserve sensitive links
\cite{hay2007anonymizing}, e.g., a data publisher can randomly remove real links and insert fake links
into an original network so as to preserve the sensitive links from being identified.
Zheleva \emph{et al.}~\cite{zheleva2008preserving} assumed that the adversary has an accurate
probabilistic model for link prediction, and they proposed several heuristic approaches to anonymizing
network data. Ying \emph{et al.}~\cite{Ying2008Randomizing} investigated the relationship between the
level of link randomization and the possibility to infer the presence of a link in a network. Further,
Ying \emph{et al.}~\cite{ying2009link} investigated the effect of link randomization on protecting
privacy of sensitive links, and they found that similarity indices can be utilized by adversaries to
significantly improve the accuracy in predicting sensitive links.

Differing from the above approaches, Fard \textit{et al.}~\cite{fard2012limiting} assumed that all
links in a network are sensitive, and they proposed to apply subgraph-wise perturbations onto a
directed network, which randomize the destination of a link within some subnetworks thereby limiting
the link disclosure. Furthermore, they proposed neighborhood randomization to probabilistically randomize
the destination of a link within a local neighborhood on a network~\cite{fard2015neighborhood}.
It should be noted that both subnetwork-wise perturbation and neighborhood randomization perturb
every link in the network based on a certain probability.

As discussed above, link prediction can be applied to predict the potential relationship between
two individuals. From another perspective, it may also increase the risk of link disclosure.
Even if the data owner removes sensitive links from the published network dataset, it may still be
disclosed by link prediction and consequently lead to privacy breach. Michael \emph{et al.}~\cite{fire2013links}
presented a \emph{link reconstruction attack}, which is a method that attacker can use link prediction
to infer a user's connections to others with high accuracy, but they did not mention how to defend
the so-called link-reconstruction attack. Naturally, one can consider finding a way to prevent the
link-prediction attack. Since link-reconstruction attack or link-prediction-based attack aims to
find out some real but unobservable links, the defense of link-prediction-based attacks is also
\textbf{target-directed}, which means that one has to preserve the targeted links from being predicted.
In the literature, most existing approaches on link prediction are based on the similarity between
pairwise nodes under the assumption that the more similar a pair of nodes are, the more likely
a link exists between them.

In sociology, \emph{triadic closure} is a popular concept to explain the evolution of a social network,
i.e., if the triad of three individuals is not closed, then the person connected to the other two
individuals would like to close this triad in order to achieve a closure in the relationship
network~\cite{Rapoport1953Spread}. Here, \emph{triadic closure} breeds several neighborhood-based
indices to measure the proximity of pairwise nodes. Zhou \textit{et al.}~\cite{zhou2009predicting}
compared several neighborhood-based similarity indices on disparate networks, and they found that
\textit{Resource Allocation} (RA) index behaves best. Recently, L\"u \textit{et al.} ~\cite{L2010Link}
published a survey on link prediction. In general, since many link prediction algorithms are designed
based on network structures, a data owner can add perturbations into the original network to reduce
the risk of targeted-link disclosure due to link-prediction-based attacks.

Note that, since a major goal of publishing social network data is to pursue useful and worthy
research, one needs to limit the level of link perturbation. This reveals that the key of defending
link-prediction-based attacks depends on how to find \textbf{Optimal Link Perturbation (OLP)},
which takes the defense effect and the level of perturbation into account, to invalid link prediction
algorithms, so that the removed sensitive links are hard to be correctly predicted. However, to the
best of our knowledge, there are very few studies on this perspective of conducting target defense
against link-prediction-based attacks. In this paper, the proposed approach is to randomly select some
links as sensitive ones, assuming that the adversary employs the RA index to predict the missing links.
Several approaches are proposed, including heuristic and evolutionary perturbations, to defend the
attacks. It will be further shown that this kind of defense can also serve as a counterpart of link
prediction to evaluate the robustness or vulnerability of various link prediction algorithms.

The main contributions of this paper are summarized as follows:
\begin{itemize}
\item A novel target defense method against link-prediction-based attacks is proposed to hide the
sensitive links from being predicted.
\item Techniques using heuristic as well as evolutionary perturbations to defend link-prediction-based
attacks are developed. With evolutionary perturbations, a new fitness function is designed, taking into
consideration of both precision and AUC (area under curve), which can reduce the calculation of
fitness values by computing the variation after perturbations. The experimental results demonstrate
the superiority of evolutionary perturbations, especially using an EDA (estimation of distribution
algorithm), and it is found that the transfer effect is closely related to the fitness.
\end{itemize}

The rest of the paper is organized as follows. In Sec.~\ref{sec:1}, the network model and some
standard metrics about link prediction are introduced. The proposed target defense against
link-prediction-based attacks is discussed. In Sec.~\ref{sec:2}, several methods are developed,
including heuristic and evolutionary perturbations. In Sec.~\ref{sec:3}, empirical experiments are
conducted to examine the proposed methods, and the experimental results are analyzed in detail.
In Sec.~\ref{sec:4}, the paper is concluded with some discussions on future directions of research.

\section{Network model and performance metrics}\label{sec:1}

\subsection{Network Model}

An undirected network is presented by $G(V, E)$, where $V$ and $E$ denote the sets of nodes and links,
respectively. Define the set of all node pairs in the network by $\Omega=\{(i,j)|\{i,j\}\in V,i\ne j\}$.
In an undirected network, link $(i,j)$ is considered the same as link $(j,i)$ and $\Omega$ contains
$\binom{k}{2} = \frac{{|V|\left( {|V| - 1} \right)}}{2}$ possible links. All observable links $E$ are
divided into two groups: the training set $E^T$ and the validation set $E^V$, where $E^T\cup E^V=E$ and
$E^T\cap E^V=\emptyset$. Furthermore, define the set of unknown node pairs and the set of non-existent
node pairs by ${U} = \Omega  - {E^T}$ and ${N} = \Omega  - E$, respectively. Obviously, ${U} = {N} +E^V$.

\subsection{Link Prediction}

In link prediction, the objective is to predict $E^V$ within $U$ utilizing the structure information
provided by $E^T$.
There are several metrics to measure the similarity of node pairs, among which local indices are
computationally efficient and well-performed. According to~\cite{zhou2009predicting} and~\cite{Fu2018Link},
the RA index stands out from several local similarity-based indices. It is defined as follows:
\begin{equation}
RA_{xy} = \sum\limits_{z \in \Gamma \left( x \right) \cap \Gamma \left( y \right)} {\frac{1}{{{d_z}}}}\,,
\end{equation}
where $\Gamma(x)$ denotes the one-hop neighbors of node $x$ and $d_z$ denotes the degree of node $z$.
To reduce the RA value of a certain node pair, one can decrease the number of common neighbors between them
or increase the degree of their common neighbors.

In this paper, two standard metrics are used for measuring the performance of link prediction, i.e.,
precision and AUC (the area under the receiver operating characteristic curve), which are defined as follows:
\begin{itemize}	
	\item \textbf{Precision}
	\\ Precision is the ratio of real missing links to predicted links. To calculate it, first sort node pairs in $U$ in
descending order according to their similarity values, as $\hat U$, and then the top-$k$ node pairs
${\hat U}^k $ in ${\hat U} $ are chosen as the predicted ones. The definition of precision is
	\begin{equation}
	precision = \frac{|{\hat U}^k \cap E^V|}{|E^V|}\,,
	\end{equation}
where $k =|E^V| $ will be used in this paper.
	\item \textbf{AUC}
	\\ AUC  measures the probability that a random missing link in $E^V$ is given a higher similarity value
than a random non-existent node pair in ${N}$. To calculate it, choose each missing link in $E^V$ and each
non-existent link in ${N}$ to compare their values: if there are $n_>$ times where a missing link has a
higher value and $n_=$ times their values are the same, then the AUC is computed as follows:
	\begin{equation}
	AUC = \frac{n_> + 0.5n_=}{n}\,,
	\end{equation}
where $n = |E^V|*|{N}|$ will be used in this paper.
\end{itemize}

\subsection{Target Defense in Link Prediction Attack}

A new target defense method is proposed here against link-prediction-based attacks. For instance, as
shown in Fig.~\ref{fig:link_privacy}, simply deleting the sensitive link $(i,j)$ cannot make it hidden,
because the value of $RA_{ij}$ is equal to 1, which is much higher than other node pairs in the network.
The deleted link could be easily re-identified, if RA is employed to predict missing links. However,
if one deletes link $(i,q)$ and meanwhile inserts link $(p,q)$, then $RA_{ij}=0.5$, which is similar to
other node pairs.

 \begin{figure}[h]
	\centering
	\includegraphics[width=1\linewidth]{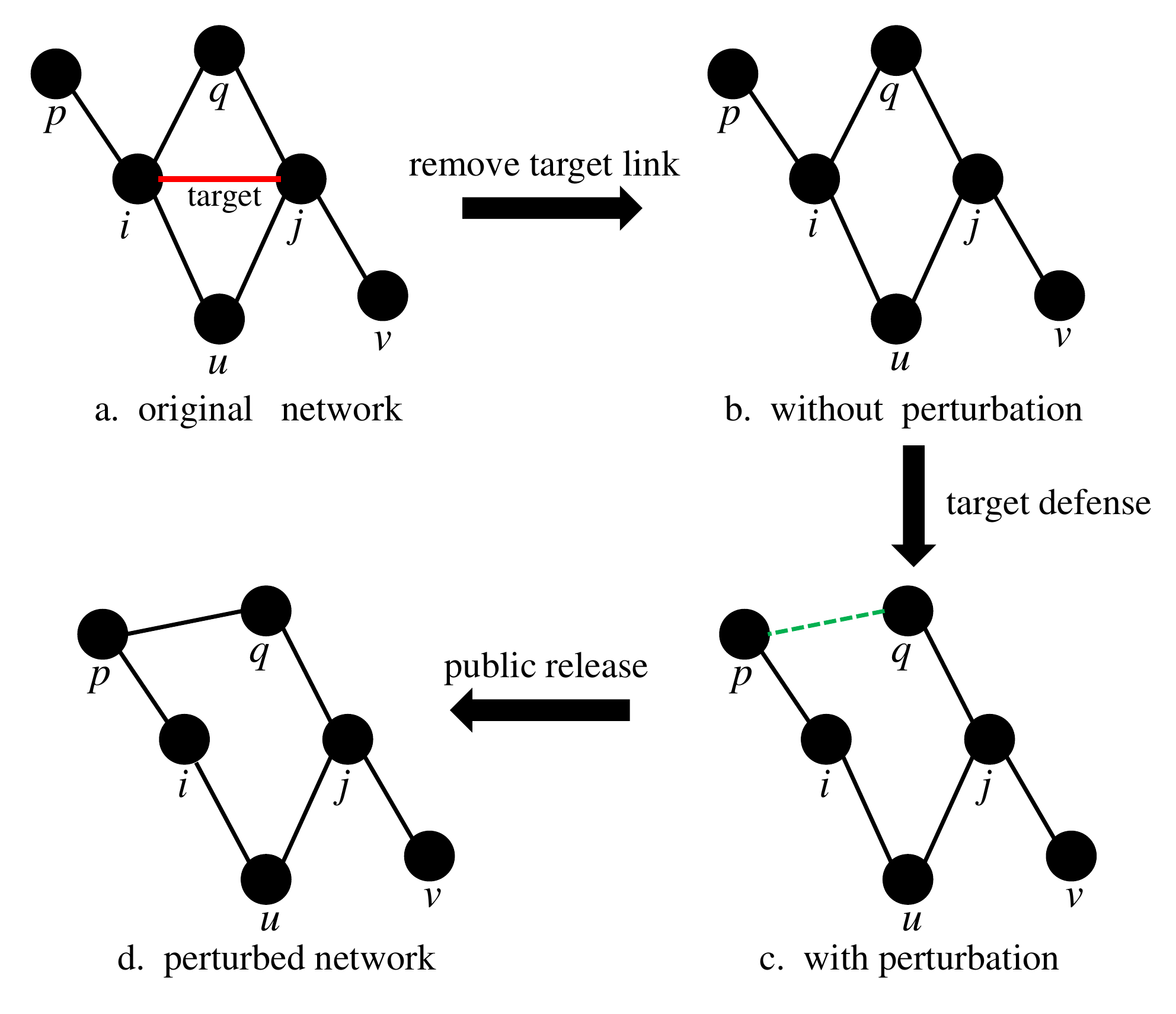}
	\caption{The whole process of defending link-prediction-based attacks via link perturbations.
             Red link $(i,j)$ is sensitive, which needs to be hidden, and green dashed link $(p,q)$
             is the inserted link. The released network can preserve the sensitive link $(i, j)$ from
             being identified by the RA index, to some extent.}
	\label{fig:link_privacy}
\end{figure}

In a network $G(V,E)$, the set $E$ is divided into two disjoint groups: $E^T$ and $E^V$. Take $E^V$ as
the set of sensitive links. The set of link perturbations $\tilde E$ is added into $E^T$, which means
deleting $E_{del}$ from $E^T$ and then inserting $E_{add}$ into $E^T$ so as to reduce the rate that
links in the set of sensitive links being re-identified by a link prediction algorithm $\mathcal{L}$,
which consequently leads to the risk of privacy leakage.

Although the defensive effect may be improved as the number of modified links increases, the number of
modifications should be minimized, taking into account the cost of perturbations and the utility of the
perturbed network. Therefore, the following two restricted conditions are considered:
\begin{itemize}
 \item the numbers of deleted and inserted links should be identical to make the total number of
 links unchanged.
 \item the numbers of deleted and inserted links should be sparse to ensure data utility.
\end{itemize}

Ut supra, target defense against link-prediction-based attacks can be described mathematically as follows:
\begin{equation}
\begin{aligned}
&\mathop {\textbf{min}}\limits_{\tilde E}\,\,\,\, \mathcal{L}(E^T + {\tilde E},E^V)\\
&\textbf{ s.t.} \quad
\begin{cases}
|\tilde E| = |E_{add}|+|E_{del}|=2 \cdot |E_{del}| \ll |E^T|\\
E_{del}\subset  E^T ,\,\,\,\, E_{add} \subset  N\,.
\end{cases}
\end{aligned}
\end{equation}
After being perturbed, $E^T$, $U$ and $N$ of the original network are changed to $\tilde E^T$,
$\tilde U$ and $\tilde N$, respectively. It then follows that
\begin{eqnarray}
&{\tilde E^T}&= E^T - E_{del} + E_{add},\\
&{\tilde U}& = \Omega - (E^T+\tilde E)= U + E_{del} - E_{add},\\
&{\tilde N}&= \Omega - (E^T+\tilde E)- E^V = N + E_{del} - E_{add}\,.
\end{eqnarray}

The relationship of each link set after a perturbation is shown in Fig. \ref{fig:link set},
where the set of non-existent node pairs before and after the perturbation are shown in grey.

\begin{figure}[h]
	\centering
	\includegraphics[width=1\linewidth]{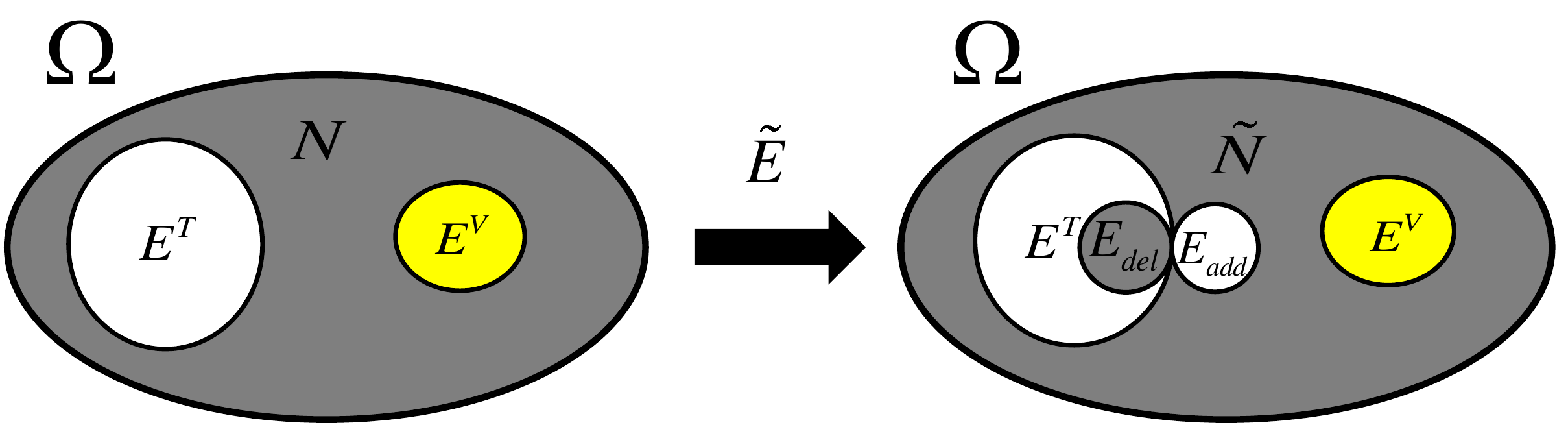}
	\caption{The venn diagram of each node-pair set after adding perturbation $\tilde E$ into the
             original network.}
	\label{fig:link set}
\end{figure}

\section{Methods}\label{section:indices}\label{sec:2}	

In this paper, assume that the adversary uses the RA index to predict the missing links.
The objective here is to preserve $E^V$ from being predicted via adding link perturbations.
Then, random, heuristic and evolutionary perturbations are applied, respectively, as
discussed below.

\subsection{Random Perturbations}

\subsubsection{Randomly Link Rewiring (RLR)}

For a given network in which sensitive links are removed, one can randomly delete some links
$\ell_{del}$ from $E^T$ and then insert the same number of new links $\ell_{add}$, which exist
in $N$.

\subsubsection{Randomly Link Swapping (RLS)}

Another common practice for randomizing a network is link swapping, which not only keeps the
total number of links unchanged but also preserves the degrees of nodes~\cite{Xuan2009Optimal}.
As shown in Fig.~\ref{swap}, link swapping removes two randomly chosen links from $E^T$ and
creates two new links existing in $N$ before swapping.

\begin{figure}[h]
	\centering
	\includegraphics[width=0.7\linewidth]{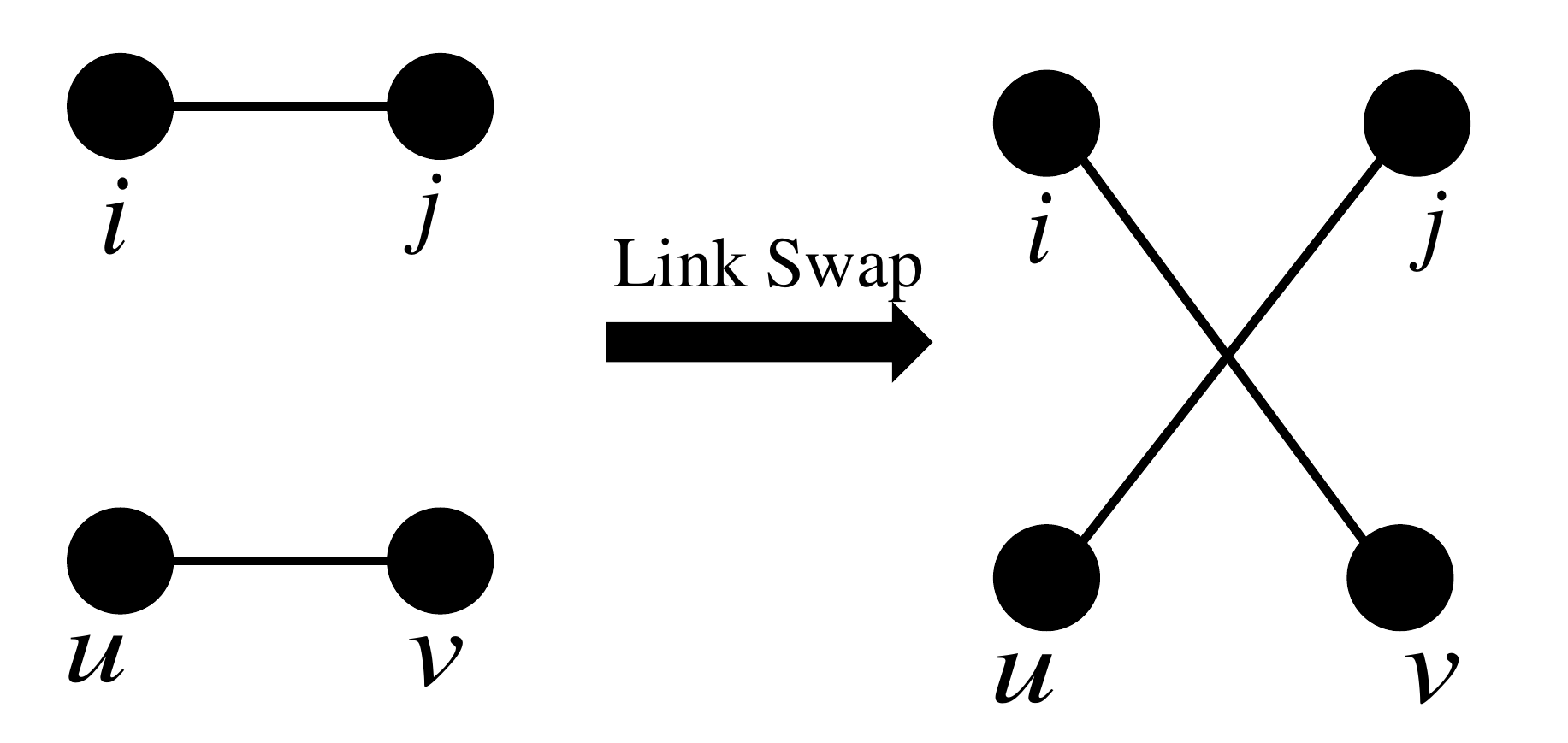}
	\caption{An example for link swapping. If either $(i,v)$ or $(i,v)$ already exists in $E^T$,
             no swapping is performed, and another attempt is made to find a suitable link pair.}
	\label{swap}
\end{figure}

\subsection{Heuristic Perturbation }

To hide the links in $E^V$, one can degrade the performance of link prediction through decreasing
the similarity values of node pairs in $E^V$ and then increasing the values of non-existent node
pairs. Here, a greedy strategy is adopted to rewire links. First, the RA index is used to calculate
the similarity values of node pairs $(i,j) \in \Omega$ and then sort them in descending order
according to their values.

For each node pair $(i,j) \in \Omega$, there are three cases, i.e., $(i,j) \in E^T$, $(i,j) \in  N$,
and $(i,j) \in E^V$. One can traverse the ordered node pairs and conduct different operations by
deleting or inserting links for each case.
\begin{itemize}
\item $(i,j) \in E^T$

If the current node pair $(i,j)$ exists in $E^T$, directly delete the link. Because the deleted
$(i,j)$ will exist in $\tilde U$ and $(i,j)$ has a high RA value.

\item $(i,j) \in N$

If the current node pair $(i,j)$ exists in $N$, select the node $k$ whose degree is the smallest
in the one-hop neighborhoods of $i$ and $j$, except their common neighbors. Then, insert $(i,k)$
or $(j,k)$ to increase the RA value of the current node pair $(i,j)$.

\item $(i,j) \in E^V$

If the current node pair $(i,j)$ exists in $E^V$, select the common neighbor $k$ whose degree is
the smallest and then delete $(i,k)$ or $(j,k)$; or select two common neighbors $k$ and $l$ whose
degrees are among the smallest. Then, insert $(k,l)$ to reduce the RA value of the node pair $(i,j)$.
 \end{itemize}
The above operation is shown in Fig.~\ref{rewire}.

It should ensure that all deleted links $\ell_{del}\in E^V$ and inserted links $\ell_{add}\in N$.
The intact pseudo-codes of the algorithms are described in Algorithms ~\ref{alg:slr}, ~\ref{alg:dl}
and ~\ref{alg:al}.

\begin{figure}[h]
	\centering
	\includegraphics[width=0.8\linewidth]{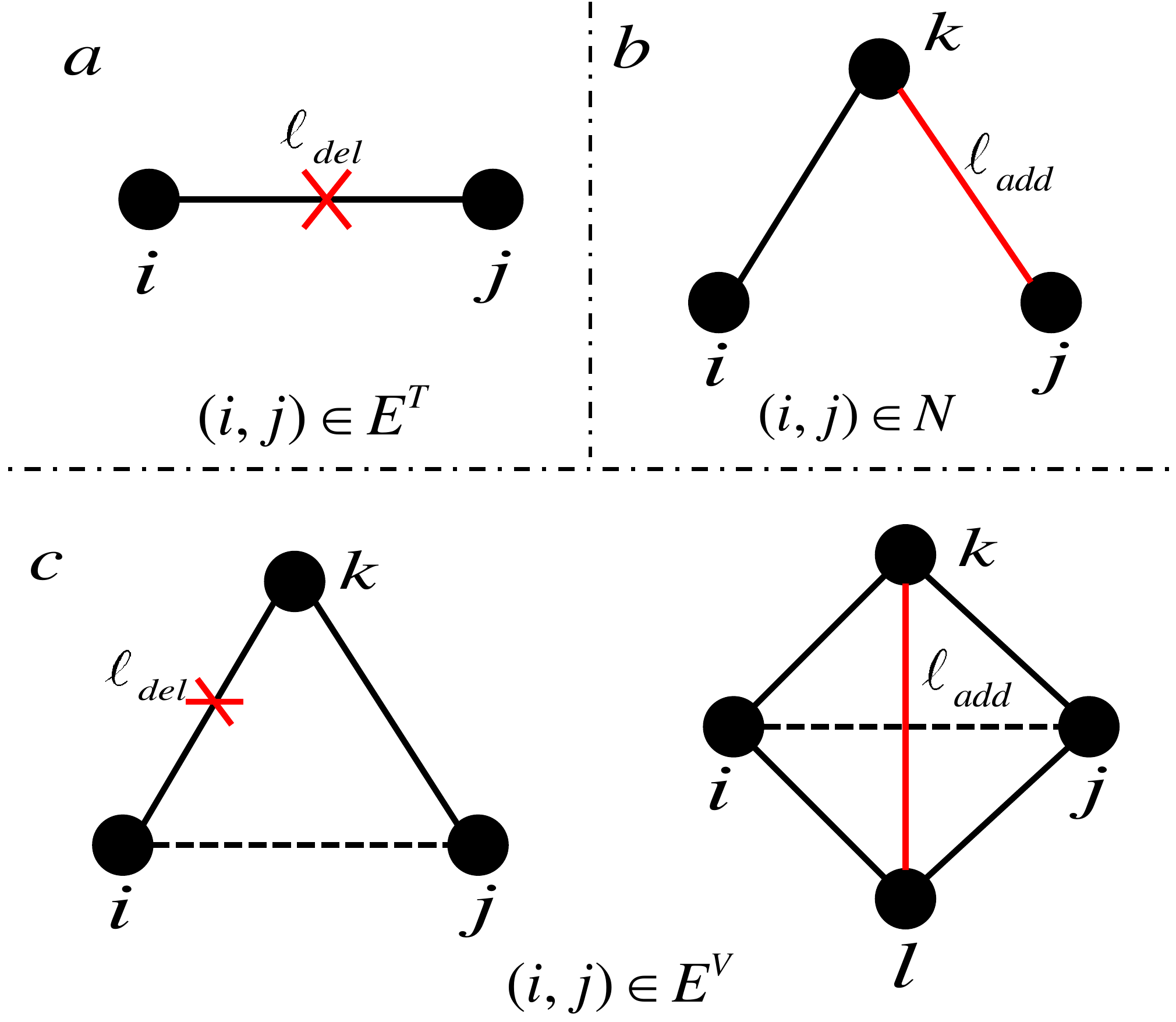}
	\caption{Four different operations through deleting or inserting links for $(i,j)$ in
             corresponding situations. }
	\label{rewire}
\end{figure}

\begin{algorithm}
\caption{  \textbf {Heuristic Perturbation (HP)}}
 \label{alg:slr}
\KwIn{$E^T$, $E^V$,  $\emph N$,  $\Omega$  and $m$.}
\KwOut {$E_{del}$ and $E_{add}$.}
Initialize $E_{del} $ and $E_{add}$\;
sort $\Omega$ in descending order to get $\hat \Omega$\;
\For {each $(i,j) \in \hat \Omega$}
{
    \If {$(i,j) \in E^T \cup E^V $}
    {
       \If {$|E_{del}| < m$}
         {$\ell_{del}$ = \textbf{\emph {Delete Link($(i,j)$)}}\;
          $E_{del}$ = $E_{del}  \cup  \ell_{del}$  \;
          \textbf{continue}\;
         }
    }
{     \If {$(i,j) \in E^V \cup N $}
        {\If {$|E_{add}| < m$ }
             {$\ell_{add}$ = \textbf{\emph {Add Link($(i,j)$)}}\;
             $E_{add}$ = $E_{add}  \cup  \ell_{add}$;}

        }
}
\If {$|E_{del}| = m$ and $|E_{add}| = m$ }
     {\textbf{Break}\;}
}
\Return $E_{del}$ and $E_{add}$.
\end{algorithm}
\begin{algorithm}
\caption{ \textbf{ \emph {Delete Link} }}
 \label{alg:dl}
\KwIn{ node pair $(i,j)$}
\KwOut{$\ell_{del}$}
\If {$(i,j)\in E^T$}
{
$E^T = E^T - (i,j)$\;
\Return $(i,j)$\;
}
 \If {$(i,j)\in E^V$}
{Find  $k=\mathop {\arg\min }\limits_k\{d_k| k \in \Gamma \left( i \right)
\cap \Gamma \left( j \right)\}$\;
Find $l = \mathop {\arg \max }\limits_l \{ {d_l}|l \in \{ i,j\} \} $\;
\If {$(k,l)$ exists}
{
    $E^T = E^T - (k,l)$ \;
    \Return $(k,l)$\;
}
}
\end{algorithm}

\begin{algorithm}
\caption{  \textbf{\emph {Add Link} }}
 \label{alg:al}
\KwIn{ node pair $(i,j)$}
\KwOut{$\ell_{add}$}

\If {$(i,j) \in E^V$}
{
    Find $k=\mathop {\arg\min }\limits_k\{d_k| k \in \Gamma \left( i \right)
    \cap \Gamma \left( j \right)\}$\;
     Find $l=\mathop {\arg\min }\limits_l\{d_l| l \in \{\Gamma \left( i \right)
     \cap \Gamma \left( j \right)-k\}\}$\;
     \If {$(k,l)$ exists and $(k,l) \notin E^V$}
     {
         $E^T$ = $E^T  \cup  (k,l)$\;
         \Return $(k,l)$\;
      }
}
\If {$(i,j)   \in N$}
{Find  $k=$
$\mathop {\arg\min }\limits_k \{d_k| k \in \{ \Gamma \left( i \right) \cup
\Gamma \left( j \right) - \Gamma \left( i \right) \cap \Gamma \left( j \right) \}\}$\;
\If {$k$ exists}
{
$E^T$ = $E^T \cup   \{(k,l)| (k,l)\notin E^V, l\in \{i,j\}\}$\;
\Return $\{(k,l)| (k,l)\notin E^V, l\in \{i,j\}\}$\;
}
}
\end{algorithm}

\subsection{Evolutionary Perturbations}

Link perturbation can be treated as a combinatorial optimization problem and the number of candidate
combinations of deleted/inserted links is $C_{|E^T|}^m \cdot C_{|N|}^m$, where $m$ is the number of
deleted/inserted links.

To hide sensitive links in $E^V$, one needs to reduce the possibility that these links can be predicted.
Both AUC and precision are indicators for evaluating the performance of a link prediction algorithm.
AUC evaluates the global performance while precision focuses on pairs with highest similarity values.
Therefore, an evolutionary algorithm is designed here to find OLP, considering
both AUC and precision together as the reduction objective.

Next, consider two evolutionary algorithms, i.e. Genetic Algorithm (GA) and Estimation of Distribution
Algorithm (EDA). First, design the common chromosome, fitness and selection operation.

\begin{itemize}
\item \textbf{Chromosome}

The chromosome consists of two parts: $\ell_{del}^i \in E^T$ and $\ell_{add}^i \in N$, $i=1,2,...,m$.
The number of deleted and inserted links are identical according to the assumption. The diagram of
chromosome is shown in Fig.~\ref{fig:gene}, where the length of chromosome depends on the experimental
setup.

\item \textbf{Fitness}

Two major link prediction performance measures, i.e. precision and AUC, are considered in the fitness
design.
Note that directly taking both precision and AUC as fitness is time-consuming. Thus, a new fitness
function is designed to simplify the calculation, as follows:
\begin{equation}
\begin{aligned}
\mathop {\textbf{\emph{{max}}}}\limits_{\ell  \in {E^V},n \in \tilde N}
&Fitness =  \alpha \cdot {\sum\limits_{n \in \tilde N} {\delta (R{A_n} >
\mathop {\max }\limits_{\ell  \in {E^V}} \{R{A_\ell }\})} }\\
&+( \frac{1}{{|\tilde N|}}\sum\limits_{n \in \tilde N} {R{A_n}}
-\frac{1}{{|{E^V}|}}\sum\limits_{\ell  \in {E^V}} {R{A_\ell }})\,,
\end{aligned}
\end{equation}
where $\delta(x)$ is an indicator function: $\delta(x)=1$ when $x$ is true; otherwise, $\delta(x) = 0$,
and $\alpha$ is a tunable parameter.
The fitness function consists of two parts: (1) denoting the number of non-existent links with higher
similarity values than the most predictable sensitive pairs, which has greater influence on the precision;
(2) denoting the difference in the average similarity values between non-existent pairs and sensitive pairs,
which affects the AUC more.

\begin{figure}[!t]
	\centering
	\includegraphics[width=0.8\linewidth]{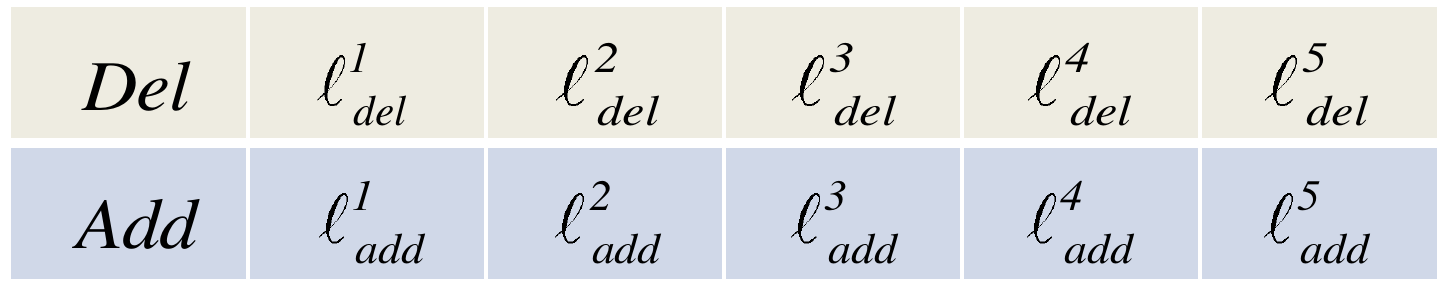}
	\caption{The diagram of chromosome in GA and EDA, respectively. It consists of two parts, i.e.
             deleted links $\ell_{del} \in E^T$ and inserted links $\ell_{add} \in N$.}
	\label{fig:gene}
\end{figure}

\item \textbf{Selection Operation}

The selection operation is conducted on roulette. To make the fitness values positive, apply exponential transform $(e^{x})$ to the fitness.
At the same time, retain $n_{elite}$ elites according to the fitness values.

\item \textbf{Mutation Operation}

Select $n_{mutation}$ chromosomes based on the fitness values for mutation operation. Then, traverse
each link in a chromosome and conduct mutation operation to it according to a mutation rate $pm$.
Specifically, one randomly replaces the deleted link $\ell_{del}$ with another one from $E^T$, and
randomly replace the inserted link $\ell_{add}$ with another one in $N$. If it happens to encounter
collision, that is, there exist duplicate links in the chromosome, then repeat the above operation
until collision disappears.
\end{itemize}

\subsubsection{Genetic Algorithm (GA)}

\begin{itemize}
\item \textbf{Crossover Operation}

Single point crossover is used here and $n_{crossover}$ chromosomes are selected according to the
values of fitness and crossover rate $pc$.
If it happens to encounter collision, as shown in Fig.~\ref{crossover}, retreat the crossover
operation of the duplicate red links.

\begin{figure}[h]
	\centering
	\includegraphics[width=0.8\linewidth]{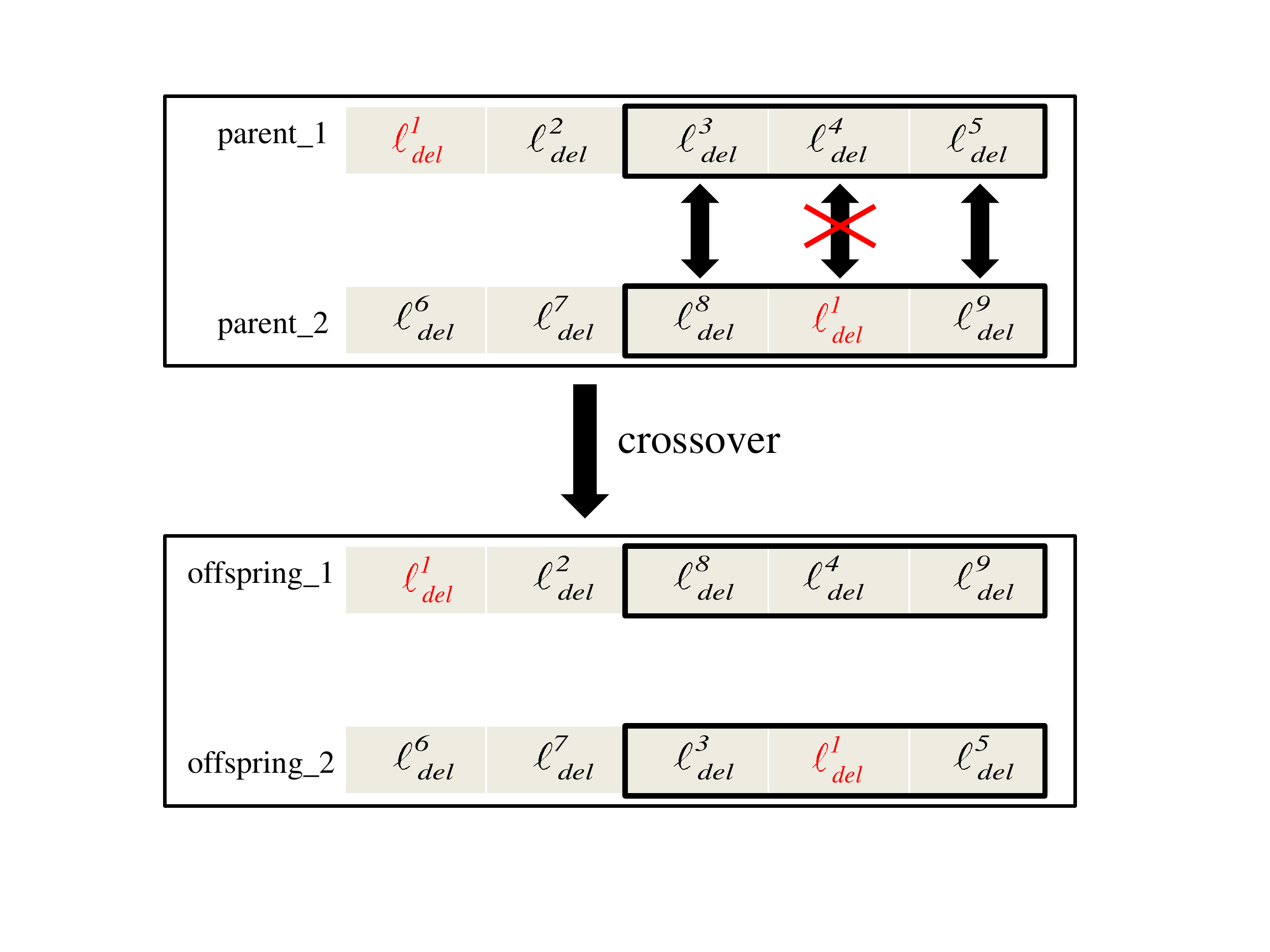}
	\caption{Crossover of chromosomes in the perspective of deleted links.
             Red genotype will encounter collision after a direct exchange.}
	\label{crossover}
\end{figure}
\end{itemize}

\begin{algorithm}[!t]
\caption{  \textbf{GA for Link Perturbation}} \label{alg:ga}
\KwIn{$m$, $n_{iteration}$, $n_{elite}$, $n_{crossover}$, $n_{mutation}$, $pc$, and $pm$.}
\KwOut{population}
Initialize population\;
\While {not convergence or not reach to $n_{iteration}$}
      {calculate the fitness of population\;
      retain $n_{elite}$ $elites$\;
     $Cro$ = \textbf{\emph{Selection Operation($n_{crossover}$)}}\;
      $\hat {Cro}$ = \textbf{\emph {Crossover Operation($pc$, $Cro$)}}\;
      $Mut$ = \textbf{\emph{Selection Operation($n_{mutation}$)}}\;
      $\hat {Mut}$ = \textbf{\emph {Mutation Operation($pm$, $Mut$)}}\;
      population = $elites  \cup \hat {Cro}  \cup \hat {Mut}$ ;
      }
\Return population\;
\end{algorithm}

\subsubsection{Estimation of Distribution Algorithm (EDA)}

EDA is a new kind of random optimization algorithm based on statistics theory.
EDA has obvious differences with GA. As mentioned above, GA applies the crossover operation
to generate new individuals; however, EDA searches better individuals through preference
sampling and statistical learning. Specifically, one samples $n_{estimation}$ individual
chromosomes according to their fitness values and then estimates the probability distributions
of the deleted links $P(\ell_{del})$ and inserted links $P(\ell_{add})$ based on simple
statistics. Finally, one generates $n_{eda}$ chromosomes according to their respective
distributions. The pseudo-codes of GA and EDA are described in Algorithms ~\ref{alg:ga},
~\ref{alg:eda} and ~\ref{alg:GP}.

\begin{algorithm}[!t]
\caption{  \textbf{EDA for Link Perturbation }}
 \label{alg:eda}
\KwIn{$m$, $n_{iteration}$, $n_{elite}$, $n_{estimation}$, $n_{eda}$,  $n_{mutation}$, and $pm$.}
\KwOut{pop}
Initialize population\;
\While {not convergence or not reach to $n_{iteration}$}
      {calculate the fitness of population\;
      retain $n_{elite}$  $ elites $\;
      $\hat {Eda}$ = \textbf{\emph {Generate Population($n_{estimation}$, $n_{eda}$)}}\;
      $Mut$ = \textbf{\emph{Selection Operation($n_{mutation}$)}}\;
      $\hat {Mut}$ = \textbf{\emph {Mutation Operation($pm$, $Mut$)}}\;
      population =  $elites  \cup  \hat {Eda}  \cup \hat {Mut}$  ;
      }

\Return pop\;
\end{algorithm}

\begin{algorithm}[!t]
\caption{  \textbf{Generate Population  }}
 \label{alg:GP}
\KwIn{$n_{estimation}$ and $n_{eda}$.}
\KwOut{population}
Initialize population\;
$Est$ = \textbf{\emph{Selection Operation($n_{estimation}$)}}\;
Estimate the distributions of  $P({\ell_{del}})$ and $P({\ell_{del}})$
by statistical sampling, where $P({\ell_{del}})  \approx  P({\ell_{del}}|Est)$
and $P({\ell_{del}}) \approx P({\ell_{add}}|Est)$ \;
\While {not reach to $n_{eda}$}
{Sample $E_{del}$ and $E_{add}$ according to $P({\ell_{del}})$
and $P({\ell_{add}})$ respectively to generate new individual\;}
\Return population\;
\end{algorithm}

\subsubsection{Accelerating the Fitness Calculation}

Since one should calculate the fitness value for each individual at every iteration,
the speed of fitness calculation directly affects the computational efficiency of
the evolutionary algorithm. It is proposed here to accelerate the calculation of
fitness by only recalculating the increment of the original network for each individual.
Note that, for RA, deleting $(i,j)$ or inserting $(i,j)$ only affects the values of
one-hop neighborhood.

\begin{theorem}
Denote by $\tilde N^{recalc}_{ij}$ the set of node pairs whose RA values need to be recalculated
in $\tilde N $. For the RA index, by either deleting or inserting $(i,j)$, one has
$|\tilde N^{recalc}_{ij}| \le \frac{1}{2} (k_i^2 + k_j^2 + k_i + k_j)+1$.
\end{theorem}

\begin{proof}
Assume that $(u,v) \in \tilde N^{recalc}_{ij}$, and denote by $G$ the original network and by
$\Gamma (i)$ the one-hop neighbors of node $i$ in $G$.

First, consider $(i,j)\in G$, and delete $(i,j)$ from $G$.
\begin{itemize}
\item  $u=i$, $v=j$

Before deleting $(i,j)$, one has $(i,j) \notin N$; after deleting $(i,j)$, one has
$(i,j) \in \tilde N$. Thus, $RA_{uv}$ needs to be calculated.

\item $u\ne i$, $v\ne j$

If $i$ or $j\in \Gamma (u) \cap \Gamma (v)$ in $G$, then $RA_{uv}$ needs to be recalculated
after deleting $(i,j)$ and $(u,v) \in \{(\Gamma (i),\Gamma (i)),(\Gamma (j),\Gamma (j))\}$;
if $\{i,j\}\notin \Gamma (u) \cap \Gamma (v)$ in $G$, then $RA_{uv}$ needs not to be recalculated
after deleting $(i,j)$.

\item $u\ne i$, $v=j$

If $i \in \Gamma (u) \cap \Gamma (j)$ in $G$, then $RA_{uv}$ needs to be recalculated after
deleting $(i,j)$ and $(u,v) \in \{(\Gamma (i),j)\}$; if $i \notin \Gamma (u) \cap \Gamma (j)$
in $G$, then $RA_{uv}$ needs not to be recalculated after deleting $(i,j)$.

\item $u=i$, $v\ne j$

$(u,v) \in \{(i,\Gamma (j))$\} by symmetry.
\end{itemize}
To sum up, $(u,v) \in \{(i,j),(\Gamma (i),\Gamma (i)),(\Gamma (j),\Gamma (j)),(\Gamma (i),j),\\(i,\Gamma (j))\}$ and the same for inserting $(i,j)$. Therefore, 
\begin{eqnarray}
\left| {\tilde N_{ij}^{recalc}} \right|&\le& \frac{1}{2}{k_i}({k_i} - 1)
+ \frac{1}{2}{k_j}({k_j} - 1) + ({k_i} + {k_j}) + 1 \nonumber\\
&=&\frac{1}{2} (k_i^2 + k_j^2 + k_i + k_j)+1\,.
\end{eqnarray}
\end{proof}

An example to illustrate the influence of deleting $(i,j)$ is shown in Fig.~\ref{fig:neighbor}.
The set $\tilde N$ of non-existent node pairs in the perturbed network can be written as
\begin{eqnarray}
{\tilde N } &=& N + E_{del} - E_{add} \nonumber\\
&=& (N - {N^{recalc}_{ij}}) + {N^{recalc}_{ij}} + ({E_{del}} - {E_{add}})\,.
\end{eqnarray}
Thus, one only needs to recalculate the RA values in $N^{recalc}_{ij}$, $E_{del}$ and $E_{add}$,
to update the fitness. Consequently, the complexity of calculating the similarity values of
node pairs in the perturbed network is significantly reduced, approximately from $O( {|V{|^2}} )$
to $O( {m{\langle k\rangle }^2} )$, where $m$ is the number of deleted/inserted links and
${\langle k\rangle }$ is the average degree of $G$.

\begin{figure}[t]
	\centering
	\includegraphics[width=1\linewidth]{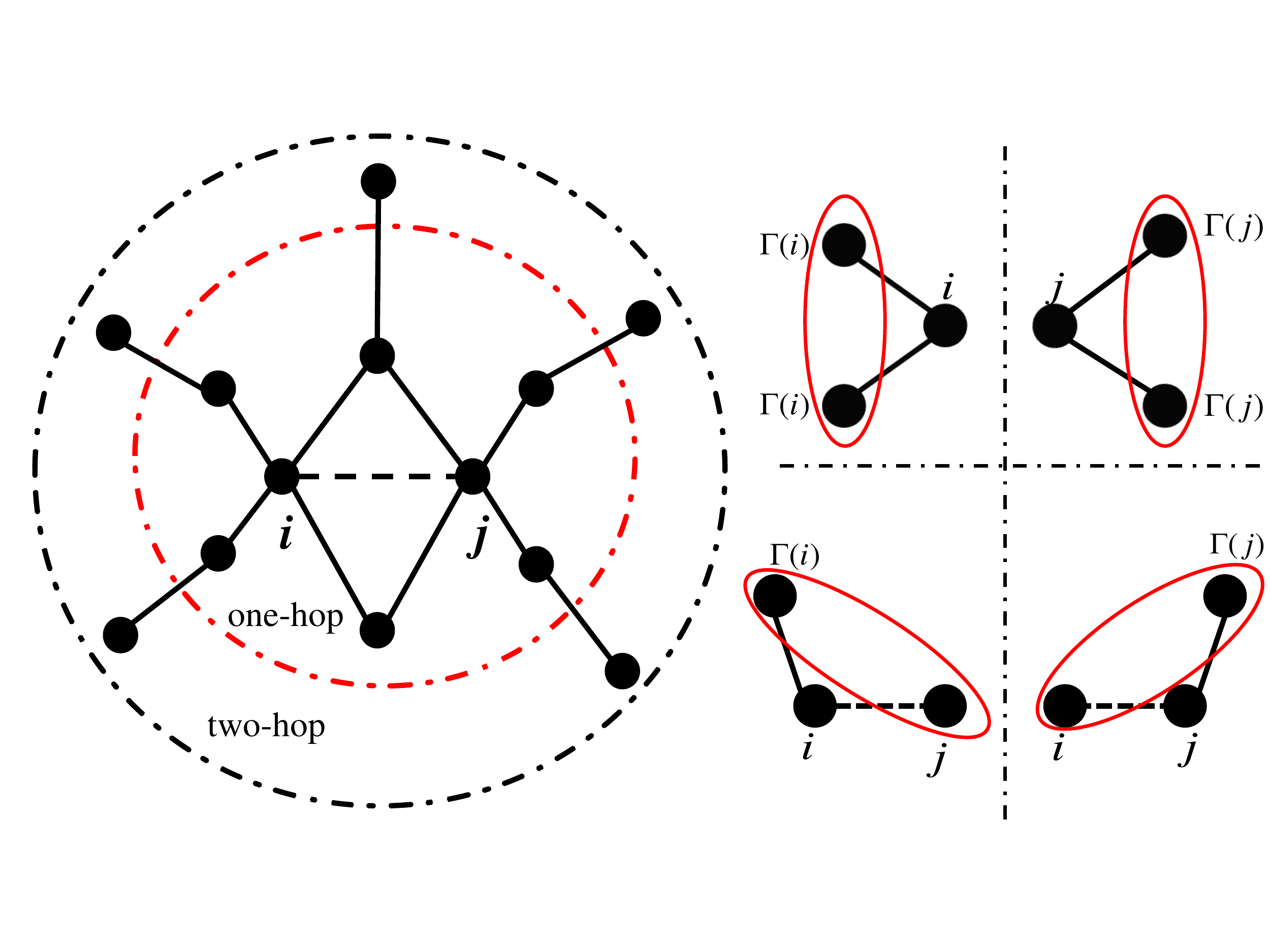}
	\caption{\textbf{Left}: neighborhood of $(i,j)$, where one-hop neighbors are within the red
             dashed circle.
             \textbf{Right}: four kinds of node pairs whose RA values need to be recalculated if
             $(i,j)$ is deleted or inserted.}
	\label{fig:neighbor}
\end{figure}

\section{Experiments and simulation results}\label{sec:3}

\begin{table}[!t]\renewcommand{\arraystretch}{1.2}
	\newcommand{\tabincell}[2]{\begin{tabular}{@{}#1@{}}#2\end{tabular}}
	\caption{Basic topological features of six networks.
             $|V|$ and $|E|$ are the numbers of nodes and edges, respectively;
             $\langle k\rangle$ is the average degree; $C$ is the clustering coefficient
             and $\langle d \rangle$ is the average distance.}
    \centering
	\begin{centering}
		\begin{tabular}{c|cccccc}
			\Xhline{1pt}
			& Mexican & Dolphin & Bomb &Lesmis& Throne& Jazz\tabularnewline
			\Xhline{1pt}
			$|V|$& \tabincell{c}{35} & 62 &64&77& 107&198\tabularnewline
			$|E|$& \tabincell{c}{117}& 159 &243&254& 352& 2742\tabularnewline
			$\langle k\rangle$& \tabincell{c}{6.686} &5.129&7.594& 6.597&6.597& 27.697\tabularnewline
			$C$& \tabincell{c}{0.448} &0.259 &0.622&0.573& 0.551& 0.617\tabularnewline
			$\langle d \rangle$ & \tabincell{c}{2.106}&3.357 & 2.691&2.641& 2.904& 2.235\tabularnewline
			\Xhline{1pt} 			
		\end{tabular}		
	\end{centering}	
	\label{topo}
\end{table}

\subsection{Data Description}

Here, six networks are selected, with topological features shown in TABLE~\ref{topo}.
\begin{itemize}
\item \textbf{Mexican political elite (Mexican)} is an undirected network, which contains the core of
the political elite including the presidents and their closest collaborators.
Links represent significant political, kinship, friendship, or business ties among them~\cite{Gil1996The}.
\item \textbf{Dolphin social network (Dolphin)} is an undirected social network of dolphins living in a
community and links present the frequent associations between pair-wise dolphins~\cite{Lusseau2003The}.
\item \textbf{Train Bombing (Bomb)} is an undirected network, which contains contacts between suspected
terrorists involved in the train bombing of Madrid on March 11, 2004, as reconstructed from newspapers.
Nodes represent terrorists and link between two terrorists means that there was a contact between
them~\cite{Hayes2006Computing}.
\item \textbf{Lesmis} is an undirected network of characters in Victor Hugo's famous novel \textit{Les Miserables}.
Nodes denote characters and two nodes are connected if the corresponding characters co-appear in the same
chapter of the book~\cite{Knuth1993The}.
\item \textbf{Game of Thrones (Throne)} is an undirected network of character interactions from the novel
\textit{ A Storm of Swords}, where nodes denote the characters in the novel and a link denotes that two
characters are mentioned together in the text~\cite{beveridge2016network}.
\item \textbf{Jazz} is a collaboration network of jazz musicians. Each node is a jazz musician and a link
denotes that two musicians have played together in a band~\cite{gleiser2003community}.
\end{itemize}

\subsection{Simulation Results}
\subsubsection{Defense Effects of Various Link Perturbations}
All links are randomly divided into 10 uniform and disjoint sets. One of the sets is selected as a
validation set $E^V$ to be the set of sensitive links that need to be protected, and the rest is used
as a training set $E^T$. Use the above-mentioned five methods including random, heuristic and evolutionary
perturbations, respectively, to make a crosswise comparison.

\begin{table}[!t]\renewcommand{\arraystretch}{1}
	\newcommand{\tabincell}[2]{\begin{tabular}{@{}#1@{}}#2\end{tabular}}
	\caption{Parameters settings for evolutionary algorithms, including GA and EDA, in each network.}
    \centering
	\begin{centering}
		\begin{tabular}{c|c|c}
			\Xhline{1.2pt}
            $Item$& Meaning&Value \tabularnewline
            \Xhline{1.2pt}
            $\alpha$& weight in fitness &- \tabularnewline
			$m$& number of deleted/inserted links&- \tabularnewline
			$n_{iteration}$& number of iterations& 1000 \tabularnewline
			$e_{elite}$&  number of retained elites& 10\tabularnewline
			$n_{crossover}$& number of chromosomes for crossover & 50  \tabularnewline
			$n_{mutation}$ & number of chromosomes for mutation & 50 \tabularnewline
            $pc$ & crossover rate&0.7 \tabularnewline
            $pm$ & mutation rate&0.1 \tabularnewline
            $n_{estimation}$ & number of chromosomes for estimation& 250\tabularnewline
            $n_{eda}$ &  number of generated population in EDA& 50\tabularnewline
			\Xhline{1.2pt}			
		\end{tabular}		
	\end{centering}	
	\label{parameter}
\end{table}

In order that all links are used for both training and validation set, a 10-fold cross-validation is used
to calculate the average precision and AUC. To ensure the sparsity of the perturbations, the proportion of
deleted and inserted links in the training set are limited to observe the downtrend of defense effect with
an increasing proportion of perturbations. The proportion is defined as the ratio of deleted/inserted links
to all links in the training set. For example, if the proportion equals 0.1, it means that the deleted and
inserted links account for 10\% of the training set, respectively. To perform a reasonable comparison, add
the perturbations using different methods in the same training set and then calculate the precision and AUC
in the same validation set.

In Random Link Rewiring (RLR), randomly delete and insert a certain proportion of links in each training set
and then repeat the procedure for one hundred times. In Randomly Link Swapping (RLS), conducting link swapping for one time means
deleting two links and adding two other links at the same time. So, one may conduct half times of link
swapping comparing to RLR and then repeat one hundred times of the procedure in each training set.
\begin{figure*}[!t]
	\centering
	\includegraphics[width=1\linewidth]{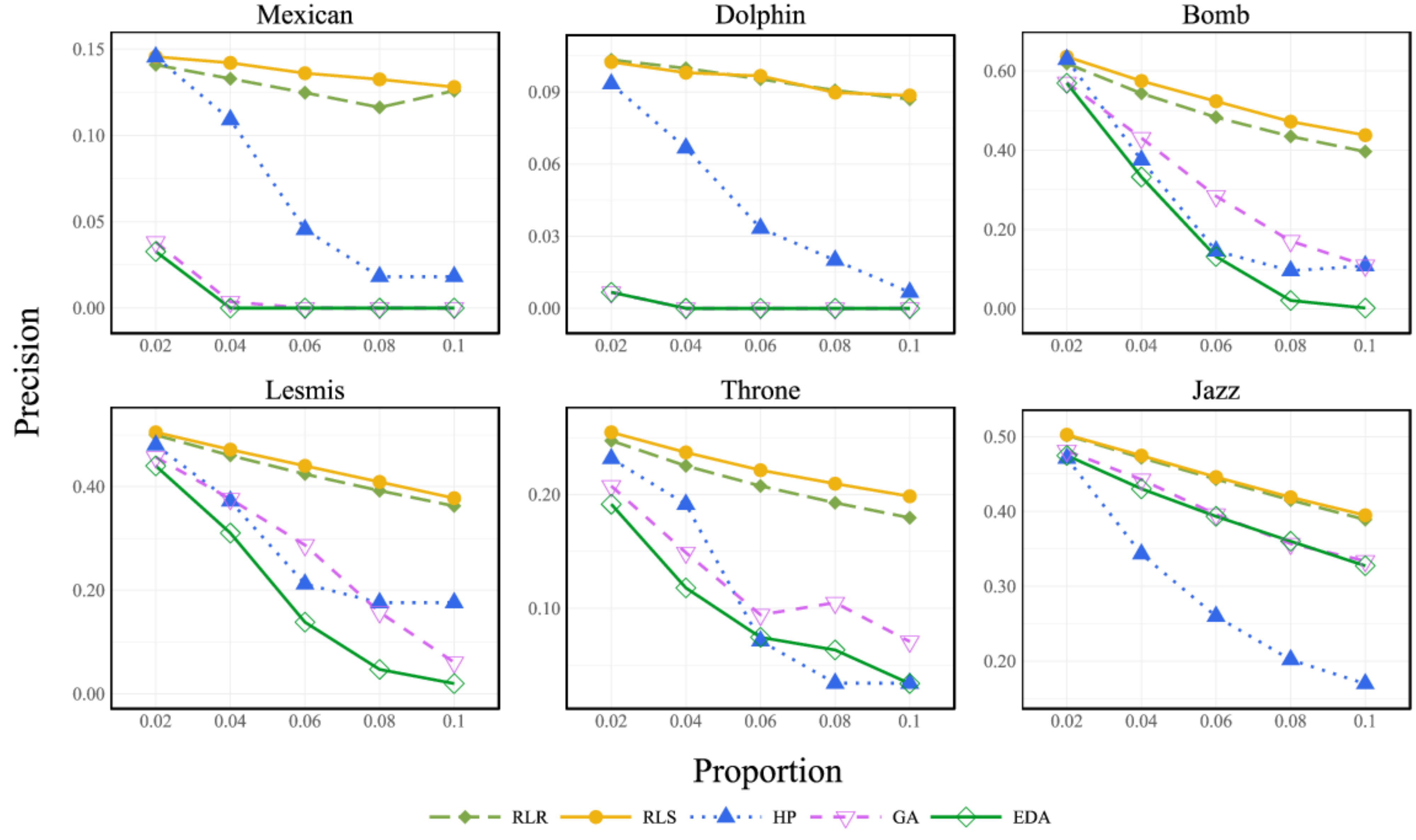}
	\caption{Average precision of each network perturbed by different methods with increasing proportions of deleted and inserted links. }
	\label{fig:PRE}
\end{figure*}
\begin{figure*}[!t]
	\centering
	\includegraphics[width=1\linewidth]{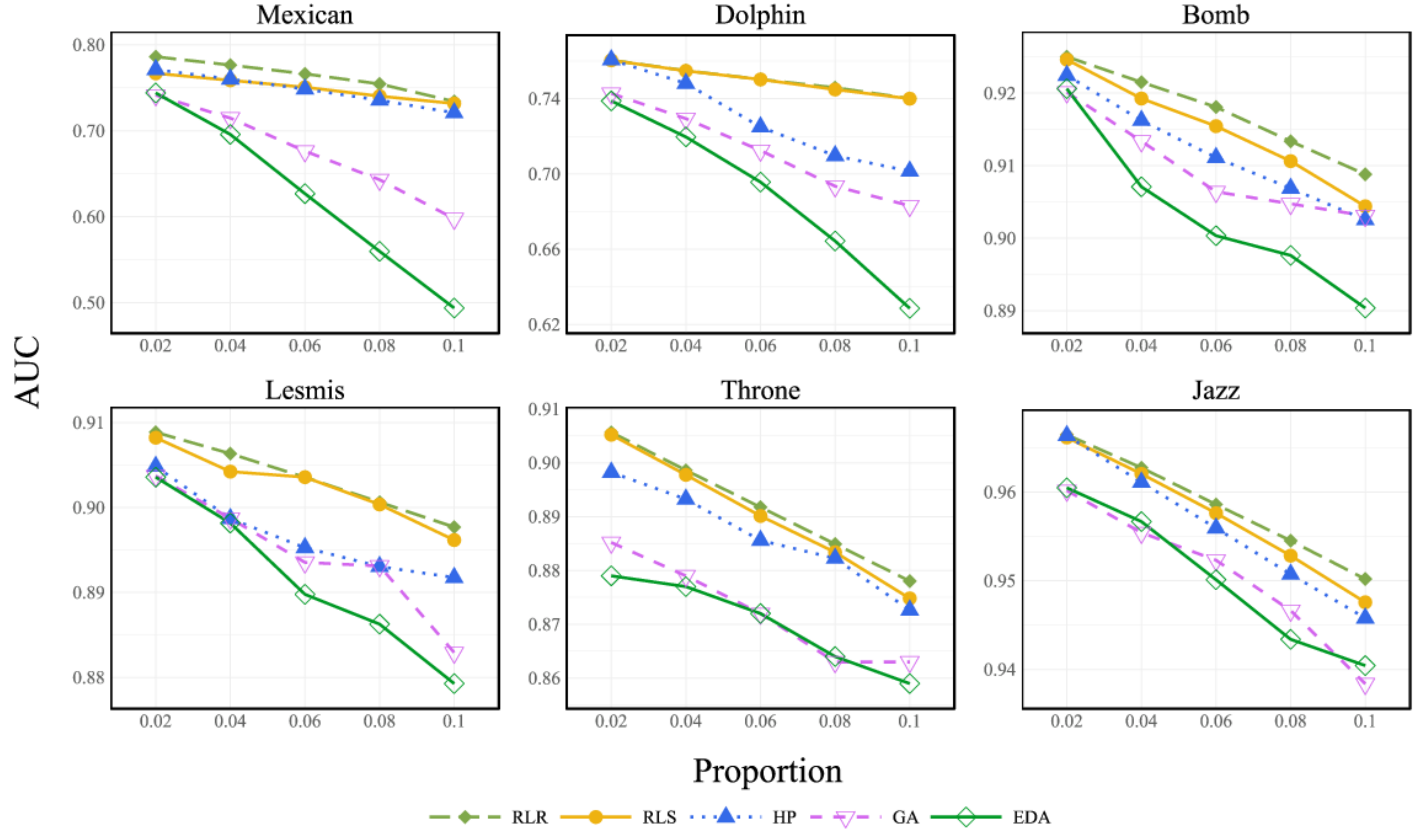}
	\caption{Average AUC of each network perturbed by different methods with increasing proportions of deleted and inserted links. }
	\label{fig:AUC}
\end{figure*}
In Heuristic Perturbation (HP), similarly conduct one hundred times of the procedure and then obtain the
average precision and AUC.
In GA and EDA,
basic parameters of the evolutionary algorithms are set to be the empirical values, as shown in
TABLE~\ref{parameter}. The main parameter to be tuned is the $\alpha$ in fitness, and the results are shown
in TABLE~\ref{alpha}. It can be seen that there does not exist any fixed $\alpha$ that is optimal for all
networks. So, one may select the value of $\alpha$ separately for each individual network.
In the experiments, for \textbf{Mexican}, \textbf{Dolphin}, \textbf{Lesmis} and \textbf{Throne}, set
$\alpha=0.01$; for \textbf{Bomb}, set $\alpha=1$; for \textbf{Jazz}, set $\alpha=0$. Then, repeat five
times in each training set and calculate the average precision and AUC of the optimal individuals in the
final population.

\begin{table*}[!ht]
\caption{Average precision and AUC in each network with different values of $\alpha$
         when the proportion of deleted and inserted equals 0.06.}
\centering
\renewcommand\arraystretch{1.1}
\begin{tabular}{cc|cccc|cccc}
\Xhline{1.2pt}
\multicolumn{2}{c|}{\multirow{2}*{\textbf{Datasets(GA)}} }& \multicolumn{4}{c|}{Precision} &\multicolumn{4}{c}{AUC} \\
\Xcline{3-10}{0.8pt}
\multicolumn{2}{c|}{}&$\alpha = 0$&$\alpha = 0.01$&$\alpha = 0.1$&$\alpha = 1$&$\alpha = 0$&$\alpha = 0.01$&$\alpha = 0.1$&$\alpha = 1$\\
\Xhline{1.2pt}
\multicolumn{2}{c|}{\textbf{Mexican}}&0.0818&\textbf{0.0364}&\textbf{0.0364}&\textbf{0.0364}&\textbf{0.708}&0.740&0.758&0.754\\
\multicolumn{2}{c|}{\textbf{Dolphin}}&0.0667&0.00667&0.00667&\textbf{0}&\textbf{0.698}&0.741&0.749&0.752\\
\multicolumn{2}{c|}{\textbf{Bomb}}&0.379&0.329&\textbf{0.288}&\textbf{0.288}&\textbf{0.878}&0.899&0.908&0.910\\
\multicolumn{2}{c|}{\textbf{Lesmis}}&0.272&0.252&\textbf{0.244}&0.248&\textbf{0.872}&0.879&0.899&0.898\\
\multicolumn{2}{c|}{\textbf{Throne}}&\textbf{0.0943}&0.0971&0.111&0.114&\textbf{0.835}&0.863&0.869&0.876\\
\multicolumn{2}{c|}{\textbf{Jazz}}&\textbf{0.397}&0.422&0.439&0.417&\textbf{0.952}&0.956&0.958&0.958\\
\Xhline{1.2pt}
\multicolumn{2}{c|}{\multirow{2}*{\textbf{Datasets(EDA)}} }& \multicolumn{4}{c|}{Precision} &\multicolumn{4}{c}{AUC} \\
\Xcline{3-10}{0.8pt}
\multicolumn{2}{c|}{}&$\alpha = 0$&$\alpha = 0.01$&$\alpha = 0.1$&$\alpha = 1$&$\alpha = 0$&$\alpha = 0.01$&$\alpha = 0.1$&$\alpha = 1$\\
\Xhline{1.2pt}
\multicolumn{2}{c|}{\textbf{Mexican}}&0.0727&\textbf{0.0273}&\textbf{0.0273}&0.0364&\textbf{0.701}&0.745&0.744&0.744\\
\multicolumn{2}{c|}{\textbf{Dolphin}}&0.0533&0.00667&0.00667&\textbf{0}&\textbf{0.689}&0.736&0.748&0.752\\
\multicolumn{2}{c|}{\textbf{Bomb}}&0.392&0.317&0.225&\textbf{0.129}&\textbf{0.867}&0.897&0.909&0.905\\
\multicolumn{2}{c|}{\textbf{Lesmis}}&0.276&0.224&0.156&\textbf{0.0680}&\textbf{0.859}&0.894&0.889&0.889\\
\multicolumn{2}{c|}{\textbf{Throne}}&0.0886&0.0771&\textbf{0.0314}&0.0343&\textbf{0.816}&0.884&0.879&0.869\\
\multicolumn{2}{c|}{\textbf{Jazz}}&\textbf{0.401}&0.434&0.429&0.425&\textbf{0.951}&0.959&0.957&0.955\\
\Xhline{1.2pt}
\end{tabular}
\label{alpha}
\end{table*}

The final results are shown in Fig.~\ref{fig:PRE} and Fig.~\ref{fig:AUC}, from which it is found that
evolutionary perturbations, especially those obtained by EDA, are superior to RLR, RLS and HP on most
network datasets.
The simulation results also show that the effect of HP is getting better as the proportion of perturbations
increases. Especially, in a larger scale network with a larger validation set $E^V$, HP even outperforms
evolutionary methods when measured in precision. Larger-scale perturbations in larger-scale networks mean
that the search space is larger. Correspondingly, increasing the scale of the evolutionary method, e.g.
increasing the numbers of individuals and iterations, may yield better results. However, time consumption
can also exceed the allowable tolerance. Clearly, improving evolutionary methods or introducing parallel
computing may reduce time consumption, which will be left for future studies.

Moreover, the running time before and after accelerating the fitness calculation has also been compared.
In Fig.~\ref{fig:time}, experiments on GA and EDA with the same parameters in the two different networks
indicate that the computational efficiency is increased significantly after acceleration.

\begin{figure}[!t]
	\centering
	\includegraphics[width=1\linewidth]{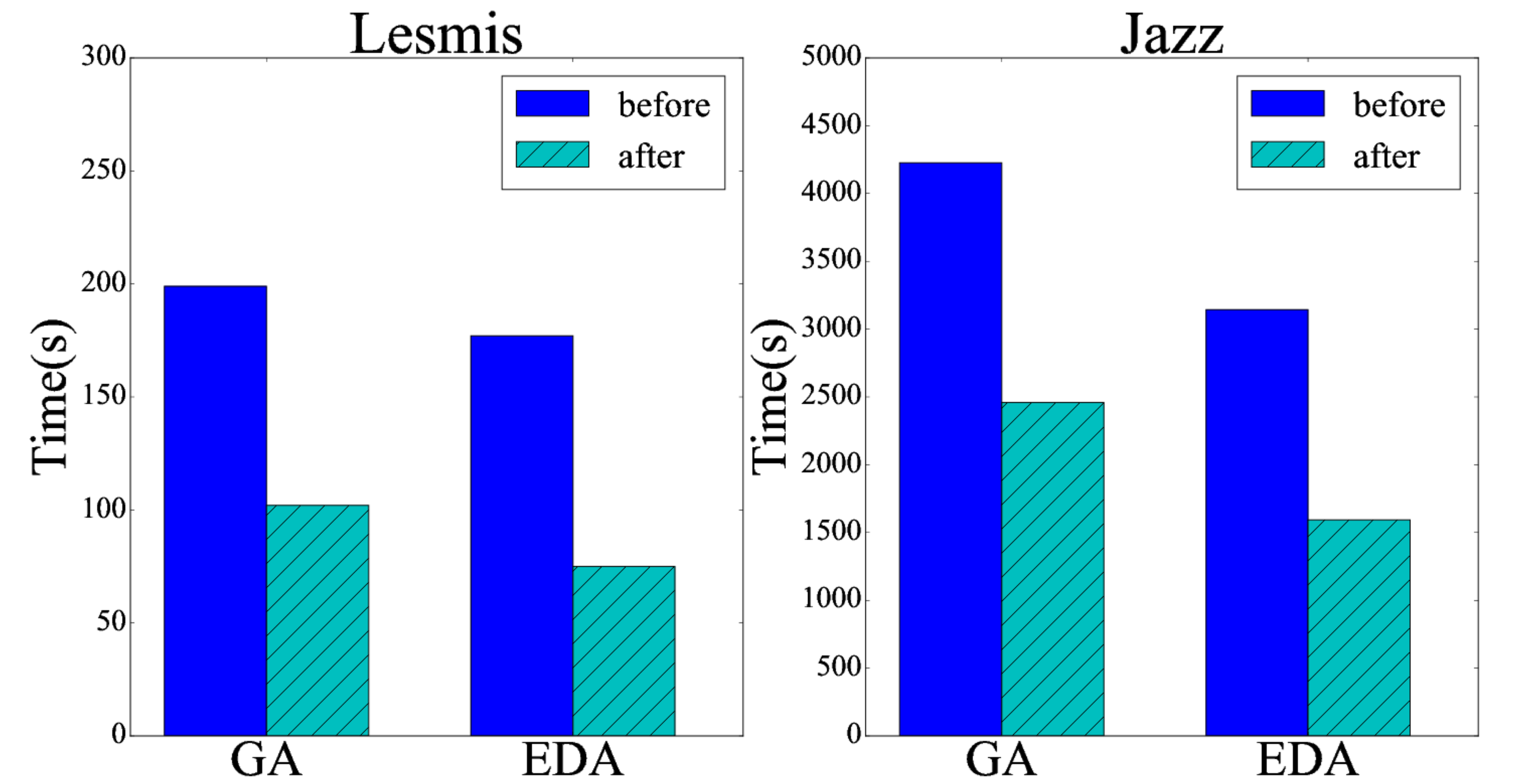}
	\caption{Average running time of GA and EDA with the same parameters in two networks
             before and after accelerating the fitness calculation.}
	\label{fig:time}
\end{figure}

\subsubsection{Transferability of Evolutionary Perturbation}

Although the above algorithms defend against specific link-prediction-based attacks (RA), an adversary
may use other methods to do link prediction. Hence, it is desirable to know the transferability of
the evolutionary perturbation especially acquired from EDA by checking the performance of other link prediction
algorithms. For this purpose, the final perturbed networks obtained with the maximum perturbation
generated by EDA together with HP are retained for each network, for which five different similarity indices are selected
as test indices, which are Common Neighbors (CN), Jaccard, Preferential Attachment (PA), Adamic-Adar (AA)
and Local Path (LP)~\cite{L2010Link}, to calculate the average precision and AUC of the perturbed networks.

The results are summarized in TABLE~\ref{transference}. One can see that the transferability of EDA is
better, no matter whether it is measured by precision or by AUC, comparing to the random perturbation (RLR/RLS) and heuristic perturbation (HP),
in most networks. However, it can also be found that, in some cases, the transfer effect of EDA is inferior
even than RLS when checked by LP, which may be due to the fact that LP considers
the higher-order similarity of node pairs while the fitness of EDA only considers one-hop neighbors,
which limits its transfer effect.

\begin{table*}[!ht]
\caption{Transferability of perturbations generated by HP and EDA measured by different similarity
indices for each network. HP (RA) means HP is based on RA and EDA (RA) means that the fitness of EDA is based on RA as well.}
\centering
\renewcommand\arraystretch{1.1}
\begin{tabular}{c|cccccc|ccccc}
\Xhline{1.2pt}
\multicolumn{2}{c|}{ \multirow{2}*{$\textbf{Mexican}$} }& \multicolumn{5}{c|}{Precision} &\multicolumn{5}{c}{AUC} \\
\Xcline{3-12}{0.8pt}
\multicolumn{2}{c|}{}&original&RLR&RLS&HP(RA)&EDA(RA)&original&RLR&RLS&HP(RA)&EDA(RA)\\
\Xhline{1.2pt}
\multicolumn{2}{c|}{RA}&0.155&0.125&0.128&0.0182&\textbf{0}&0.777&0.734&0.732&0.721&\textbf{0.495}\\
\multicolumn{2}{c|}{CN}&0.118&0.0993&0.106&0.0273&\textbf{0}&0.760&0.720&0.715&0.702&\textbf{0.516}\\
\multicolumn{2}{c|}{Jaccard}&0.109&0.0876&0.0796&0.0273&\textbf{0}&0.752&0.709&0.701&0.693&\textbf{0.496}\\
\multicolumn{2}{c|}{PA}&0.0546&0.0677&0.0596&0.0364&\textbf{0.0300}&0.625&0.616&0.625&0.588&\textbf{0.557}\\
\multicolumn{2}{c|}{AA}&0.127&0.120&0.119&0.0182&\textbf{0}&0.776&0.733&0.731&0.721&\textbf{0.504}\\
\multicolumn{2}{c|}{LP($\alpha=0.5$)}&0.127&0.0896&0.105&0.0545&\textbf{0.0165}&0.729&0.692&0.682&0.684&\textbf{0.564}\\
\Xhline{1.2pt}
\multicolumn{2}{c|}{ \multirow{2}*{$\textbf{Dolphin}$} }& \multicolumn{5}{c|}{Precision} &\multicolumn{5}{c}{AUC} \\
\Xcline{3-12}{0.8pt}
\multicolumn{2}{c|}{}&original&RLR&RLS&HP(RA)&EDA(RA)&original&RLR&RLS&HP(RA)&EDA(RA)\\
\Xhline{1.2pt}
\multicolumn{2}{c|}{RA}&0.107&0.0865&0.0885&0.00667&\textbf{0}&0.765&0.740&0.740&0.701&\textbf{0.629}\\
\multicolumn{2}{c|}{CN}&0.113&0.0953&0.0954&0.0333&\textbf{0}&0.760&0.736&0.736&0.700&\textbf{0.648}\\
\multicolumn{2}{c|}{Jaccard}&0.113&0.0920&0.0973&0.0200&\textbf{0.00455}&0.760&0.737&0.736&0.699&\textbf{0.651}\\
\multicolumn{2}{c|}{PA}&0.0133&0.0121&0.0132&\textbf{0}&\textbf{0}&0.623&0.616&0.623&0.589&\textbf{0.583}\\
\multicolumn{2}{c|}{AA}&0.133&0.0978&0.0999&0.0133&\textbf{0}&0.765&0.741&0.741&0.702&\textbf{0.632}\\
\multicolumn{2}{c|}{LP($\alpha=0.5$)}&0.120&0.107&0.103&0.0533&\textbf{0.0226}&0.792&0.768&0.762&0.772&\textbf{0.711}\\
\Xhline{1.2pt}
\multicolumn{2}{c|}{ \multirow{2}*{$\textbf{Bomb}$} }& \multicolumn{5}{c|}{Precision} &\multicolumn{5}{c}{AUC} \\
\Xcline{3-12}{0.8pt}
\multicolumn{2}{c|}{}&original&RLR&RLS&HP(RA)&EDA(RA)&original&RLR&RLS&HP(RA)&EDA(RA)\\
\Xhline{1.2pt}
\multicolumn{2}{c|}{RA}&0.713&0.395&0.438&0.104&\textbf{0.00182}&0.929&0.909&0.904&0.903&\textbf{0.891}\\
\multicolumn{2}{c|}{CN}&0.571&0.452&0.483&\textbf{0.271}&0.317&0.915&0.897&0.888&0.891&\textbf{0.883}\\
\multicolumn{2}{c|}{Jaccard}&0.475&0.359&0.327&0.292&\textbf{0.274}&0.912&0.894&0.886&0.876&\textbf{0.874}\\
\multicolumn{2}{c|}{PA}&0.229&0.186&0.219&\textbf{0.138}&0.139&0.774&0.767&0.774&0.749&\textbf{0.744}\\
\multicolumn{2}{c|}{AA}&0.658&0.459&0.498&0.200&\textbf{0.197}&0.926&0.907&0.901&0.902&\textbf{0.891}\\
\multicolumn{2}{c|}{LP($\alpha=0.5$)}&0.488&0.379&0.389&0.404&\textbf{0.325}&0.880&0.864&\textbf{0.848}&0.868&0.851\\
\Xhline{1.2pt}
\multicolumn{2}{c|}{ \multirow{2}*{$\textbf{Lesmis}$} }& \multicolumn{5}{c|}{Precision} &\multicolumn{5}{c}{AUC} \\
\Xcline{3-12}{0.8pt}
\multicolumn{2}{c|}{}&original&RLR&RLS&HP(RA)&EDA(RA)&original&RLR&RLS&HP(RA)&EDA(RA)\\
\Xhline{1.2pt}
\multicolumn{2}{c|}{RA}&0.540&0.365&0.378&0.180&\textbf{0.0199}&0.914&0.898&0.896&0.891&\textbf{0.879}\\
\multicolumn{2}{c|}{CN}&0.484&0.359&0.357&0.204&\textbf{0.182}&0.906&0.895&0.888&\textbf{0.881}&\textbf{0.881}\\
\multicolumn{2}{c|}{Jaccard}&0.0360&0.202&\textbf{0.0606}&0.0920&0.114&0.914&0.870&0.859&0.854&\textbf{0.850}\\
\multicolumn{2}{c|}{PA}&0.104&0.0936&0.102&0.0800&\textbf{0.0768}&0.782&0.780&0.782&\textbf{0.768}&0.777\\
\multicolumn{2}{c|}{AA}&0.524&0.380&0.378&0.192&\textbf{0.111}&0.912&0.898&0.893&0.888&\textbf{0.886}\\
\multicolumn{2}{c|}{LP($\alpha=0.5$)}&0.376&0.311&0.300&0.316&\textbf{0.236}&0.875&0.871&\textbf{0.856}&0.867&0.869\\
\Xhline{1.2pt}
\multicolumn{2}{c|}{ \multirow{2}*{$\textbf{Throne}$} }& \multicolumn{5}{c|}{Precision} &\multicolumn{5}{c}{AUC} \\
\Xcline{3-12}{0.8pt}
\multicolumn{2}{c|}{}&original&RLR&RLS&HP(RA)&EDA(RA)&original&RLR&RLS&HP(RA)&EDA(RA)\\
\Xhline{1.2pt}
\multicolumn{2}{c|}{RA}&0.274&0.180&0.198&0.0343&\textbf{0.0291}&0.912&0.878&0.875&0.873&\textbf{0.859}\\
\multicolumn{2}{c|}{CN}&0.200&0.174&0.181&0.109&\textbf{0.103}&0.892&0.865&0.859&0.849&\textbf{0.847}\\
\multicolumn{2}{c|}{Jaccard}&0.0600&0.0553&0.0414&0.0543&\textbf{0.0357}&0.861&0.837&0.825&\textbf{0.813}&\textbf{0.813}\\
\multicolumn{2}{c|}{PA}&0.123&0.110&0.122&0.100&\textbf{0.0929}&0.769&0.767&0.769&0.764&\textbf{0.754}\\
\multicolumn{2}{c|}{AA}&0.240&0.192&0.203&\textbf{0.0486}&0.0814&0.908&0.878&0.873&0.869&\textbf{0.859}\\
\multicolumn{2}{c|}{LP($\alpha=0.5$)}&0.163&0.145&0.157&0.117&\textbf{0.116}&0.867&0.852&\textbf{0.838}&0.860&0.845\\
\Xhline{1.2pt}
\multicolumn{2}{c|}{ \multirow{2}*{$\textbf{Jazz}$} }& \multicolumn{5}{c|}{Precision} &\multicolumn{5}{c}{AUC} \\
\Xcline{3-12}{0.8pt}
\multicolumn{2}{c|}{}&original&RLR&RLS&HP(RA)&EDA(RA)&original&RLR&RLS&HP(RA)&EDA(RA)\\
\Xhline{1.2pt}
\multicolumn{2}{c|}{RA}&0.512&0.390&0.395&\textbf{0.172}&0.327&0.960&0.950&0.948&0.946&\textbf{0.940}\\
\multicolumn{2}{c|}{CN}&0.497&0.382&0.378&\textbf{0.237}&0.335&0.954&0.938&0.930&0.933&\textbf{0.929}\\
\multicolumn{2}{c|}{Jaccard}&0.512&0.390&0.398&\textbf{0.275}&0.342&0.960&0.945&0.942&\textbf{0.926}&0.938\\
\multicolumn{2}{c|}{PA}&0.132&0.121&0.130&\textbf{0.106}&0.112&0.770&0.761&0.769&\textbf{0.746}&0.752\\
\multicolumn{2}{c|}{AA}&0.518&0.391&0.391&\textbf{0.220}&0.336&0.961&0.944&0.938&0.939&\textbf{0.935}\\
\multicolumn{2}{c|}{LP($\alpha=0.5$)}&0.359&0.295&0.281&0.265&\textbf{0.264}&0.908&0.889&\textbf{0.874}&0.902&0.882\\
\Xhline{1.2pt}
\end{tabular}
\label{transference}
\end{table*}

\section{Conclusions and research outlook}\label{sec:4}

In this paper, a target defense algorithm against link-prediction-based attacks is proposed
for social networks. Both heuristic and evolutionary perturbations techniques are used,
taking into account both defense effect and data utility. A special fitness function is
designed, which can simultaneously measure two link prediction indices, i.e. precision and
AUC, and can reduce computation time by calculating variations after perturbations. The
experimental results on six real-world networks show that evolutionary perturbations,
especially those obtained by EDA, outperform other baseline methods, for both precision and
AUC measures, in most cases. Finally, the transfer effect of evolutionary perturbations
generated by EDA is verified, showing that evolutionary perturbations are transferable
and can be used to defend other link-prediction-based attacks when the similarity measure
is closely related to the fitness index.

However, for large-scale networks, especially when the number of links to be hidden is very large,
the proposed algorithm does not take advantage of the network structure, consequently the
evolutionary efficiency is limited. Therefore, improving evolutionary methods or introducing
parallel computing so as to reduce computation time are good topics for future studies.


%

\bibliographystyle{IEEEtran}

\bibliography{REFs}

\end{document}